\documentclass[12pt]{article}
\usepackage{amsmath}
\usepackage{graphicx,psfrag,epsf}
\usepackage{float}
\usepackage{subfig}
\usepackage{enumerate}
\usepackage{natbib}
\usepackage{url} 
\usepackage{amssymb}
\usepackage{amsmath}
\usepackage{enumitem}  
\usepackage{geometry}
\usepackage[english]{babel}
\usepackage{amsthm}
\usepackage{xcolor}
\usepackage{verbatim}
\usepackage{booktabs}
\usepackage{bm}
\usepackage[utf8]{inputenc}
\usepackage[english]{babel}
 
\usepackage{natbib}
\newcommand{\R}{{\mathbb R}}
\newcommand{\iid}{{\stackrel{i.i.d.}{\sim}}}
\DeclareMathOperator*{\argmax}{arg\,max}
\DeclareMathOperator*{\argmin}{arg\,min}

\newtheorem{theorem}{Theorem}[section]
\newtheorem{lemma}[theorem]{Lemma}
\newtheorem{proposition}[theorem]{Proposition}

\newtheorem{observation}[theorem]{Observation}
\newtheorem{remark}{Remark}

\newcommand{\blind}{1}

\begin{document}

\def\spacingset#1{\renewcommand{\baselinestretch}%
{#1}\small\normalsize} \spacingset{1}


\if1\blind
{
  \title{\bf Central Limit Theorems for Classical Multidimensional Scaling}
  \author{Gongkai Li\thanks{
    This work is partially supported by \textit{DARPA D3M through contract FA8750-17-2-0112} } , \hspace{.2cm}
    Minh Tang, \hspace{.2cm} Nicolas Charon,\hspace{.2cm} Carey E. Priebe\\
    Johns Hopkins University, Department of Applied Math and Statistics\hspace{0.2cm}\\
    }    
  \maketitle
} \fi

\if0\blind
{
  \bigskip
  \bigskip
  \bigskip
  \begin{center}
    {\LARGE\bf Title}
\end{center}
  \medskip
} \fi

\begin{abstract}
Classical multidimensional scaling is a widely used method in dimensionality reduction and manifold learning. 
The method takes in a dissimilarity matrix and outputs a low-dimensional configuration matrix based on a spectral decomposition. 
In this paper, we present three noise models and analyze the resulting configuration matrices, or embeddings. 
In particular, we show that under each of the three noise models the resulting embedding gives rise to a central limit theorem. We also provide compelling simulations and real data illustrations of these central limit theorems. This perturbation analysis represents a significant advancement over previous results regarding classical multidimensional scaling behavior under randomness.
\end{abstract}

\noindent%

\section{Background and Overview}
\label{sec:B&O}

Inference based on dissimilarities is of fundamental importance in statistics, data mining and machine learning \citep{dissimilarityPatternRecog}, with applications ranging from neuroscience \citep{Vogelstein386} to psychology \citep{Carroll1970} and economics \citep{cmdsecon}. In each of these fields, rather than directly observing the feature values of the objects, often we observe only  the dissimilarities or ``distances" between pairs of objects (inter-point distances). A common approach to dimensionality reduction and subsequent inference problems involving dissimilarities is to embed the observed distances into some (usually Euclidean) space to recover a configuration that faithfully preserves observed distances, and then proceed to perform inference based on the resulting configuration \citep{Leeuw-Heiser, BGbook, Torgerson, Cox2008}. The popular classical multidimensional scaling 
(CMDS) dimensionality reduction method provides an example of such an embedding scheme into Euclidean space, in which we have readily available tools to perform statistical inference. Furthermore, CMDS also forms the basis for several other more recent approaches to nonlinear dimension reduction and manifold learning \citep{Scholkopf, ChenBuja}, such as Isomap \citep{Isomap} and Random Forest manifold learning \citep{CriminisiandShotton} among others. 

Although widely used, the behavior of CMDS under randomness remains largely unexplored. Several recent papers have highlighted this omission. \citet{DsquaredplusE} write ``Despite the popularity of multi-dimensional scaling, very little is known about to what extent the distances between the embedded points could faithfully reflect the true pairwise distances when observed with noise."; \citet{Fan} write ``[W]e are not aware of any statistical results measuring the performance of MDS under randomness, such as perturbation analysis when the objects are sampled from a probabilistic model." and \citet{Peterfreund&Gavish} write ``To the best of our knowledge, the literature does not offer a systematic treatment on the influence of ambient noise on MDS embedding quality." This paper addresses this acknowledged gap in the literature.

\subsection{Review of Classical Multidimensional Scaling}
\label{sec:cmds}
Given an $n \times n$ hollow symmetric dissimilarity matrix $D$, and an embedding dimension $d$, we seek $X \in \mathbb{R}^{n \times d}$, where the rows $X_1, X_2, \dots, X_n \in \mathbb{R}^d$ of $X$ represent coordinates of points in $\R^d$, such that the overall inter-point distances between $X_i$ and $X_j$ are ``as close as possible" to the distances given by the dissimilarity matrix $D$. More specifically, CMDS involves the following steps:
\begin{enumerate}
  \item Compute the matrix $B = -\frac{1}{2}P D^2 P$, where $D^2$ is $D$ matrix entry-wise squared, and $P = I - \frac{\bm{1}\bm{1}^\top}{n}$ is the double centering matrix. Here $I$ denotes the $n \times n$ identity matrix and $\bm{1} = (1, \dots, 1)^\top \in \mathbb{R}^n$.
  \item Extract the $d$ largest positive eigenvalues $s_1, \dots, s_d$ of $B$ and the corresponding eigenvectors ${u}_1, \dots, u_d$.
  \item Let $X = U_B S_B^{1/2} \in \mathbb{R}^{n \times d}$, where $U_B = (u_1, \dots, u_d)$ and $S_B = \textrm{diag}(s_1, \dots, s_{d})$. Each row of $X$ represents the coordinate of a point in $\R^d$.
\end{enumerate}
In essence, CMDS minimizes the Strain loss function defined as $L(X) := \| XX^{\top} - B \|_{F}$ where $\|\cdot\|_{F}$ denote the Frobenius norm of a matrix. Furthermore, the resulting configuration ${X}$ centers all points around the origin, resulting in an inherent issue of identifiability: $X$ is unique only up to an orthogonal transformation. In the following presentation, we will write $X = U_B S_B^{1/2} W$ where $W$ is some orthogonal matrix, for a suitably transformed $X$. 

\section{Noise Model and Embedding}
\label{sec:M&E}
In this section, we propose three different but related noise models for the matrix of observed dissimilarities.  
Suppose we have inter-point distances of $n$ points in $\mathbb{R}^d$, and the resulting distance matrix is given by $D \in \R^{n \times n}$, i.e. $D_{ij} = \|x_i - x_j\|_{2}$. Let $D^2$ denote the entry-wise square of $D$ and  $\Delta$ be the  dissimilarity matrix we observed (such as measured via a scientific experiment). We consider three error models for $\Delta$:

\subsection{Model 1: $\Delta^2 = D^2 + E$}
\label{D^2+E}
An error model proposed in \cite{DsquaredplusE} for $\Delta$ is $\Delta^2 = D^2 + E$,  where we can think of $D^2$ as the ``signal'' matrix and $E$ as the ``noise''. We shall assume that $E$ satisfies the following conditions:
\begin{enumerate}[label=(\roman*)]
  \item $\mathbb{E}[E] = 0$, hence $\mathbb{E}[\Delta^2] = D^2$.
  \item $E$ is hollow and symmetric.
  \item Entries $E_{ij}$ are independent and $\mathrm{Var}(E_{ij}) = \sigma^2$.
  \item Each $E_{ij}$ follows a sub-Gaussian distribution.
\end{enumerate}

\subsection{Model 2: $\Delta = D + E$}
\label{D+E}
Another realistic error model is $\Delta = D + E$. 
Here we also require that the random matrix $E$ satisfies conditions (i) to (iv) in section \ref{D^2+E} along with a constant third and fourth moment conditions, i.e., (v) $\mathbb{E} [E_{ij} ^ 3] \equiv \gamma$ and $\mathbb{E} [E_{ij} ^ 4] \equiv \xi$ for all $i, j$.

\subsection{Model 3: Matrix Completion}
\label{matrix_completion}
In \cite{Chatterjee}, the author developed the connection between the true distance matrix and the distance matrix with missing entries for a general metric. Restricting our attention to the Euclidean distance, we propose the following matrix completion model:
\\
Suppose with probability $q$ we observe $\Delta_{ij} = D_{ij}$ and with probability $1 - q$, $\Delta_{ij}$ is missing (in which case we set $\Delta_{ij} = 0$). Our model becomes $\Delta = D + E$ where $E_{ij}$ is a Bernoulli random variable which takes value $-D_{ij}$ with probability $1-q$ and takes value $0$ with probability $q$. It is easy to see that $\mathbb{E}[\Delta] = q \cdot D$ and $\mathbb{E}[\Delta^2] = q \cdot D^2$. 

For each of the above noise models, we apply CMDS to $\Delta$ to get the resulting configuration matrix $\hat{X}$, and use the following notations for this procedure:
\begin{enumerate}
\item Let $\hat{B} = -\frac{1}{2} P \Delta^2 P$.
\item Let $S_{\hat{B}} \in \mathbb{R}^{d \times d}$ be the diagonal matrix of $d$ largest eigenvalues of $\hat{B}$ and $U_{\hat{B}} \in \mathbb{R}^{n \times d}$ be the matrix whose orthogonal columns are the corresponding eigenvectors. 
\item The matrix $\hat{X} = U_{\hat{B}} S_{\hat{B}}^{1/2} \in \R^{n \times d}$ is the ``embedding of $\Delta$" into $\mathbb{R}^d$.  
\end{enumerate}
A natural question arises regarding  how the added noise affects the embedding configuration. That is, what is the relationship between the embedding $X$ from $D$ as in Section \ref{sec:cmds} and the embedding $\hat{X}$ from $\Delta$?

\subsection{Related Works}
\label{RW}
The problem of recovering an Euclidean distance matrix from noisy or imperfect observations of pairwise dissimilarity scores arises naturally in many different contexts. For example, in \cite{DsquaredplusE}, the authors proposed the model $\Delta^2 = D^2 + E$ and showed that there exists an estimator 
$$\hat{D}^{2}:= \argmax_{M \in \mathcal{D}^{(2)}_n} \Bigl\{ \frac{1}{2} \|\Delta^2 - M\|_{F}^2 + \lambda_n \mathrm{trace} \,\,(-\frac{1}{2}PMP) \Bigr \}$$ 
for $D^2$. Here $\mathcal{D}^2_n$ is the set of $n \times n$ {\em squared} Euclidean distance matrix and $\lambda_n$ is a tuning parameter. In particular, Corollary 6 in \cite{DsquaredplusE} states that under suitable model on E, with probability approaching to one we have 
$$\|\hat{D}^2 - D^2 \|_F^2 \leq 36n\sigma^2(r+1)$$
where $\sigma$ is the variance of the noise and $r$ is the rank of $D^{2}$. In this paper we can get, as a corollary of ours results, a bound of the same order on $\|\hat{D}^2 - D^2\|_{F}^2$. Furthermore, our central limit theorem on the configuration matrix $X$ is a more refined limiting result of a different flavor.

On the other hand, completing a distance matrix with missing entries has been a popular problem in the engineering and social sciences; see, for example, \cite{Alfakih1999,Bakonyi,Singer9507, Spence1974} and distance matrix completion is closely related to multidimensional scaling \cite{BGbook,Chatterjee, Javanmard2013, Montanari}. Especially noteworthy is Theorem~2.5 of \cite{Chatterjee}, where the author established an upper bound for the mean squared error on the estimator $\tilde{M}$ for a general distance matrix $M$. More specifically, let $(K,d)$ be a compact metric space and $x_1, x_2, \dots, x_n$ be $n$ arbitrary points in $K$. Let $M$ be the $n \times n$ matrix whose $ij$-entry is $d(x_i,x_j)$. Let $\epsilon > 0$  be such that $q \geq n^{-1 + \epsilon}$. For a given $\delta > 0$, let $N(\delta)$ be the covering number of $K$ using balls of radius $\epsilon$ with respect to the metric $d$. Then there exists an estimator $\tilde{M}$ obtained by truncating the singular value decomposition of $M$ such that
$$ \mathrm{MSE}(\tilde{M}) \leq C \inf_{\delta > 0} \min \Bigl\{  \frac{\delta + \sqrt{ N(\delta/4) /n}}{\sqrt{q}} , 1 \Bigr \} + C(\epsilon) e^{-ncq}$$
where $c$ and $C$ are constants depending on the truncation level $\eta$ for 
the singular values of $M$ and $C(\epsilon)$ is a constant depending only on $\epsilon$ and $\eta$.
Of particular interest is the application of this theorem to the Euclidean distance matrix, for which we obtain roughly
$$\textrm{MSE} (\tilde{M}) \leq \frac{Cn^{-1/3}}{\sqrt{q}}.$$ 
Another relatively new and slightly different result on the CMDS configuration matrix $X$ on the incomplete Euclidean distance matrix is given in \cite{Taghizadeh}, in which Theorem 1 states that with high probability, we have 
$$\|\hat{X} - X\|_F \leq \mathcal{O}(\frac{\sqrt{n}}{\sqrt{q}}).$$ 
Our central limit theorem in this paper improves upon both result.
In addition, the Euclidean distance matrix completion problem can also be viewed from an optimization point of view. See \cite{Tasissa2018ExactRO} for a review of such approaches.  

\section{Main Results}
\label{main}

Recall that a random variable $X$ is sub-Gaussian if $\mathbb{P}[ | X | > t] \leq 2 e^{ -\frac{t^2}{ {K}^2} }$ for some constant $K$ and for all $t \geq 0$.  Associated with a sub-Gaussian random variable is a Orlicz norm defined as $ \| X \|_{\psi_2} = \inf \{ t >0 : \mathbb{E} \exp(\frac{X^2 }{t^2}) \leq 2 \}$. 
A random vector $X$ in $\R^n$ is called sub-Gaussian if the one-dimensional marginals $\big \langle X, x \big \rangle$ are sub-Gaussian random variables for all $x \in \R^n$, and the corresponding sub-Gaussian norm of $X$ is defined as $\|X\|_{\psi_2} = \sup\limits_{x \in S^{n-1}} \| \big \langle X, x \big \rangle \|_{\psi_2}$.

\subsection{Main Theorems}
We now present central limit theorems for the rows of the CMDS configuration $\hat{X}$ for the three noise models in \S~\ref{sec:M&E}. Intuitively speaking, the theorems established that the rows of $\hat{X}$, after some orthogonal transformation, is approximately normally distributed around the rows of $X$. Furthermore, the covariance matrix will depend on the noise model and the true distribution of the points in the underlying space and are substantially different between the three noise models considered. In particular, the covariance matrix for the noise model $\Delta^2 = D^2 + E$ in Theorem~\ref{main_theorem_D^2+E} depends only on the variance $\sigma^2$ of the noise $E_{ij}$.
This is in contrast with the covariance matrices of the model $\Delta = D + E$ and the model $\mathbb{E}[\Delta] = q D$ in Theorem~\ref{main_theorem_D+E} and Theorem~\ref{main_theorem_matrix_completion}, both of which depend also on the underlying true distances $D_{ij}$. The machinery involved in proving these results are by and large the same and we refer the reader to the Appendix for detailed proofs. Finally, for ease of exposition, we denote by $(A)_i$ the $i$-th row of a matrix.

\begin{theorem}
\label{main_theorem_D^2+E} (Central Limit Theorem for CMDS of $\Delta^2 = D^2 + E$)\\
Let $Z_1, Z_2, \dots, Z_n \iid F$ for some sub-Gaussian distribution $F$ on $\R^d$. Let $D$ be the Euclidean distance matrix generated by the $Z_k$'s, i.e. $D_{ij} = \|Z_i - Z_j\|$, and suppose that $\max\limits_{1 \leq i \leq n} \sum\limits_{j=1}^{n} D_{ij}^{2} \gg \log^{4}{n}.$ Let $\Delta^2 = D^2 + E$ where the noise matrix $E$ satisfy the conditions in Section \ref{D^2+E}, i.e, (i) $\mathbb{E}[E] = \bm{0}$, (ii) $E$ is hollow and symmetric, 
(iii) the entries $E_{ij}$ are independent for $i \leq j$ with $\mathrm{Var}[E_{ij}] \equiv \sigma^2$, and (iv) each $E_{ij}$ follows a sub-Gaussian distribution. We emphasize that the $E_{ij}$ need not be identically distributed. Denote by $\hat{X}_n$ the CMDS embedding configurations of $\Delta$ into $\mathbb{R}^{d}$. Then there exists a sequence of $d \times d$ orthogonal matrices $\{W_n\}_{n=1}^{\infty}$ such that for any $\alpha \in \mathbb{R}^{d}$ and any fixed row index $i$, we have 
  $$ \lim_{n {\to} \infty} \mathbb{P} \{\sqrt{n} [(\hat{X}_n W_n)_i - (Z_i - \bar{Z}) ]\leq \alpha\} = \Phi(\alpha, \Sigma) $$
  where $\bar{Z}$ is the mean of $Z_k$'s and $\Phi(\alpha, \Sigma)$ denotes the CDF of a multivariate Gaussian with mean $0$ and covariance matrix $\Sigma$, evaluated at $\alpha$. 
  Here $\Sigma = \frac{\sigma^2}{4} {\Xi}^{-1}$ where $\Xi =  \mathrm{Cov}(Z_k) \in \R^{d \times d}$.
\end{theorem}

\begin{remark}
We can relax the common variance requirement (iii) in Theorem~\ref{main_theorem_D^2+E}. Let $\textrm{Var}(E_{ij}) = \sigma_{ij}^2$ and suppose that the collection $(D^{2}_{ij} - \Delta^{2}_{ij})(Z_{j} - \mu_{z})$s satisfy the multivariate Lindeberg-Feller condition. Define $\Sigma_{i} = \frac{1}{n} \sum\limits_{j \neq i} \sigma_{ij}^2 \textrm{Cov}(Z_k)$. We then obtain the following variant of Theorem \ref{main_theorem_D^2+E}:
$$\sqrt{n} \Sigma_{i}^{-\frac{1}{2}} [(\hat{X_n}W_n)_i - (Z_i - \bar{Z}))] \rightarrow \mathcal{N}(0, I)$$
\end{remark}

\begin{theorem} (Central Limit Theorem for CMDS of $\Delta = D + E$)\\
\label{main_theorem_D+E}
Let $Z_1, Z_2, \dots, Z_n \iid F$ for some sub-Gaussian distribution $F$ on $\R^d$. Let $D$ be the Euclidean distance matrix generated by the $Z_k$'s, i.e. $D_{ij} = \|Z_i - Z_j\|$ and suppose that $ \max\limits_{1 \leq i \leq n} \sum\limits_{j=1}^{n} D_{ij}^{2} \gg \log^{4}{n}.$ Let $\Delta = D + E$ 
and suppose that the noise matrix $E$ satisfy, in addition to the conditions in Theorem~\ref{main_theorem_D^2+E}, the condition (v) 
$\mathbb{E}[E_{ij}^3] \equiv \gamma$ and $\mathbb{E}[E_{ij}^4] \equiv \xi$. Denote by $\hat{X}_n$ the CMDS embedding configurations of $\Delta$ into $\R^d$. 
Then there exists a sequence of $d \times d$ orthogonal matrices $\{W_n\}_{n=1}^{\infty}$ such that for any $\alpha \in \mathbb{R}^{d}$ and any fixed row index $i$,  
  $$ \lim_{n {\to} \infty} \mathbb{P} \{\sqrt{n} [(\hat{X}_n W_n)_{i} - (Z_i - \bar{Z})] \leq \alpha\} = \int_{\mathrm{supp}(F)} 
  \Phi(\alpha, \Sigma(\emph{z})) dF(\emph{z})$$
  where $\bar{Z}$ is the mean of $Z_k$'s and $\Phi(\alpha, \Sigma)$ denotes the CDF of a multivariate Gaussian with mean $0$ and covariance matrix $\Sigma$, evaluated at $\alpha$. Here $\Sigma (\emph{z})= {\Xi}^{-1} \widetilde{\Sigma}(\emph{z}) {\Xi}^{-1}$ where $\Xi := \mathrm{Cov}(Z_i) \in \R^{d \times d}$ and, with $\mu_z = \mathbb{E}[Z_i] \in \mathbb{R}^{d}$,
  $$\widetilde{\Sigma} (\emph{z}) := \mathbb{E}_{Z_k}\Bigl[(\sigma^2 \|\emph{z} - Z_k\|^2 + \gamma \| z_i - Z_j \|+ \frac{1}{4} \xi - \frac{\sigma^4}{4}) (Z_k- {\mu}_{z}) (Z_k - {\mu}_{z})^\top\Bigr] $$  is a covariance matrix depending on $\emph{z}$.
\end{theorem}


\begin{theorem} (Central Limit Theorem for CMDS of $\Delta = D$ with missing entries)\\
\label{main_theorem_matrix_completion}
 Let $Z_1, Z_2, \dots, Z_n \iid F$ for some sub-Gaussian distribution $F$ on $\R^d$. Let $D$ be the Euclidean distance matrix generated by the $Z_i$'s, i.e. $D_{ij} = \|Z_i - Z_j\|$. Suppose that with probability $q_n \in [0,1]$ we observe the distance $D_{ij}$ and with probability $1 - q_n$ it is missing, i.e., $\Delta = D + E$ where $E_{ij} = (- D_{ij}) \times \mathrm{Bernoulli}(1-q_n)$. Denote by $\hat{X}_n$ the CMDS embedding configurations of $\Delta$ into $\R^d$. 
 Then there exists a sequence of $d \times d$ orthogonal matrices $\{W_n\}_{n=1}^{\infty}$ such that if $n q_n = \omega(\log^{4}{n})$, then for any $\alpha \in \mathbb{R}^{d}$ and any fixed row index $i$,  
  $$ \lim_{n {\to} \infty} \mathbb{P} \{\sqrt{n} [(\hat{X_n}W_n)_{i} - \sqrt{q_n}(Z_i - \bar{Z})] \leq \alpha\} = \int_{\mathrm{supp}(F)} \Phi(\alpha, \Sigma(\emph{z})) dF(\emph{z})$$
  where $\bar{Z}$ is the mean of $Z_i$'s and $\Phi(\alpha, \Sigma)$ denotes the CDF of a multivariate Gaussian with mean $0$ and covariance matrix $\Sigma$, evaluated at $\alpha$. Here $\Sigma (\emph{z})= {\Xi}^{-1} \widetilde{\Sigma}(\emph{z}) {\Xi}^{-1}$, $\Xi := \mathrm{Cov}(Z_i) \in \R^{d \times d}$ and with $\mu_z = \mathbb{E}[Z_i] \in \mathbb{R}^{d}$,
  $$\widetilde{\Sigma} (\emph{z}) := \mathbb{E}\Bigl[\tfrac{1-q_n}{4} \|\emph{z} - Z_k\|^4 (Z_k- {\mu}_{z}) (Z_k - {\mu}_{z})^\top\Bigr]$$ 
  is a covariance matrix depending on $\emph{z}$.
\end{theorem}

\section{Empirical Results}
\label{ER}

For illustrative purpose, we will focus on the error model $\Delta = D + E$ as in Section \ref{D+E} and Theorem \ref{main_theorem_D+E}. Experimental results for the other error models are completely analogous. 

\subsection{Three Point-mass Simulated Data}
\label{SD}

As a simple illustration of our CMDS CLT,  we embed noisy Euclidean distances obtained from $n$ points into $\R^2$.
We consider three points $x_1,x_2,x_3 \in \mathbb{R}^2$ for which the inter-point distances are 3,4 and 5
(these three points form a right triangle)
and generate $n_k = \pi_k n$ points equal to $x_k$, $k=1,2,3$, where $\pi = [0.2,0.3,0.5]^\top$.
The resulting Euclidean inter-point distance matrix $D$ is then subjected to uniform noise,
yielding $\Delta = D+E$ where $E_{ij} \iid \textrm{Uniform}(-4,+4)$ for $i<j$ and $E_{ij}=E_{ji}$.
For this case, our CLT for CMDS embedding into two dimensions gives class-conditional Gaussians.
For each $n \in \{50, 100, 500, 1000\}$, Figure \ref{fig:Simulation result} compares,
for one realization, the theoretical vs.\ estimated means and covariances matrices (95\% level curves).
Table \ref{tab:cov_1} shows the empirical covariance matrix for one of the point masses, $\hat{\Sigma}^{(1)}$, behaving in accordance with Theorem \ref{main_theorem_D+E}. 

Table \ref{tab:cov_1} investigates the empirical covariance matrix for one of the point masses, and its entry-wise variance, as a function of $n$. The theoretical covariance matrix is $\Sigma^{(1)} = \begin{bmatrix} 13.56 & -3.06 \\ -3.06 & 22.65 \end{bmatrix}$.\\

{
\centering
\label{tab:cov_1}
\begin{tabular}{ c c c c c c  }
\toprule
\textbf{}  & \textbf{$n$=50} & \textbf{$n$=100} & \textbf{$n$=500} & \textbf{$n$=1000}\\
\midrule\\
\addlinespace[-2ex]
$\hat{\Sigma}^{(1)}:$ &
$  \begin{bmatrix}  14.15 & 0.25 \\ 0.25 &  79.07 \end{bmatrix}$ &
$  \begin{bmatrix} 13.67 & -0.79 \\ -0.79 & 98.96 \end{bmatrix}$ &
$ \begin{bmatrix}  13.65 & -2.34 \\ -2.34 &  41.02 \end{bmatrix}$ &
$ \begin{bmatrix} 13.63 & -2.70 \\ -2.70 & 31.76 \end{bmatrix}$ &\\
\addlinespace[2ex]
$\mathrm{Var}\begin{bmatrix}
           \hat{\Sigma}^{(1)}_{11} \\
           \hat{\Sigma}^{(1)}_{12} \\
           \hat{\Sigma}^{(1)}_{22} \end{bmatrix}:$ &
$  \begin{bmatrix}  41.25 \\ 113.31 \\  829.52 \end{bmatrix}$&
$  \begin{bmatrix} 19.29  \\ 68.06 \\ 984.45 \end{bmatrix}$ &
$ \begin{bmatrix}  3.67 \ \\ 7.87 \\  31.71 \end{bmatrix}$ &
$ \begin{bmatrix} 1.71  \\ 3.25 \\ 11.08 \end{bmatrix}$ &\\
\addlinespace[2ex]
\bottomrule
\end{tabular}
\captionof{table}{Empirical average of covariance matrix $\hat{\Sigma}^{(1)}$, and entry-wise variance, via 500 simulations.} 
}
\begin{remark}
In this simulation we relax the requirement that the entries of $\Delta$ should be nonnegative in order to illustrate the phenomenon of decreasing covariance with increasing $n$. 
\end{remark}

\begin{figure}[htbp]
    \centering
    \subfloat[$n$=50]{{\includegraphics[width=.35\textwidth]{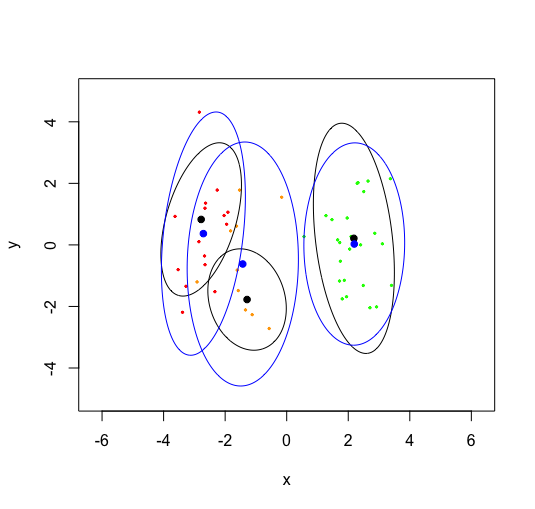} }}
    \quad
    \subfloat[$n$=100]{{\includegraphics[width=.35\textwidth]{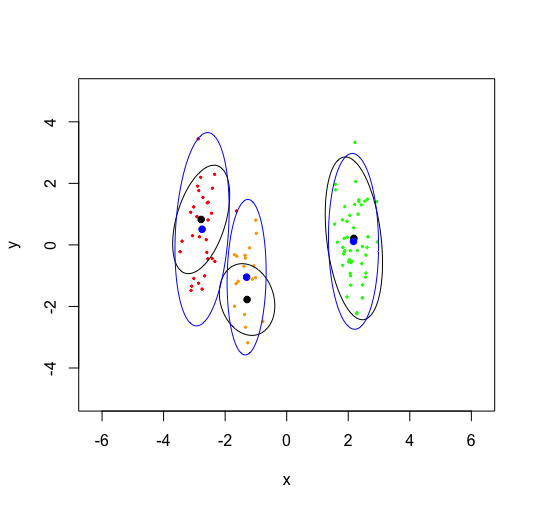} }}
    \quad
    \subfloat[$n$=500]{{\includegraphics[width=.35\textwidth]{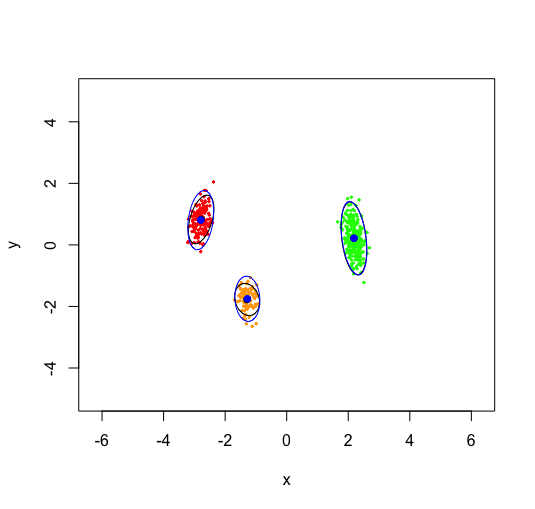} }}
    \quad
    \subfloat[$n$=1000]{{\includegraphics[width=.35\textwidth]{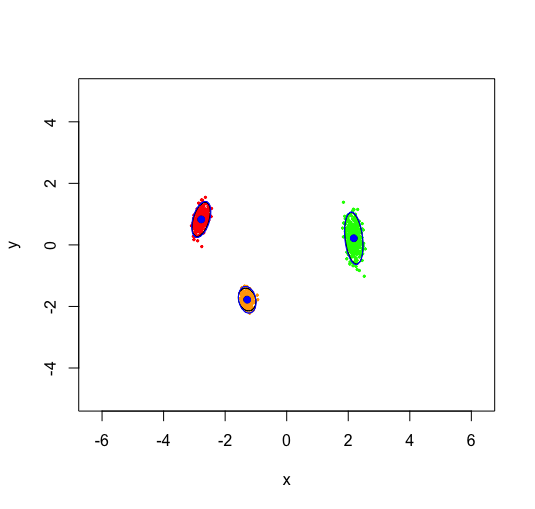} }}
    \caption{Simulation results for $n$=50, 100, 500 and 1000 points, as described in Section \ref{SD}. The blue ellipses are the 95\% level curves of the empirical covariance matrix, and the blue dots are the empirical centers for three classes. The black dots are the true positions of $x_1$, $x_2$ and $x_3$, and the black ellipses are the 95\% level curve for the theoretical covariance matrices as in Theorem \ref{main_theorem_D+E}. Note that the blue and black centers and ellipses coincide for large $n$.}%
    \label{fig:Simulation result}%
\end{figure}

\subsection{Shape clustering}
As a second illustration of the effect of noise on CMDS, we examine a more involved clustering experiment in the (non-Euclidean) shape space of closed curves. In this experiment, we consider boundary curves obtained from silhouettes of the Kimia shape database. Specifically, we restrict attention to three predefined classes of objects (bottle, bone, and wrench) and take from each class three different examples of shapes all given by planar closed polygonal curves representing the objects' outline. Figure \ref{fig:bottle_bone_wrench} shows one instance for each of the bottle, bone, and wrench class. A database of noisy curves is then created as follows: for each of the nine template shapes, we generate 100 noisy realizations in which vertices of the curve are moved along the curve's normal vectors with random distances drawn from independent Gaussian distributions at each vertex. This results in a total of 900 noisy versions of the initial curves such as the ones displayed in Figure \ref{fig:noisy_bottle_bone_wrench}.

\begin{figure}[htp]
    \centering
    \subfloat[Bottle]{{\includegraphics[width=.25\textwidth]{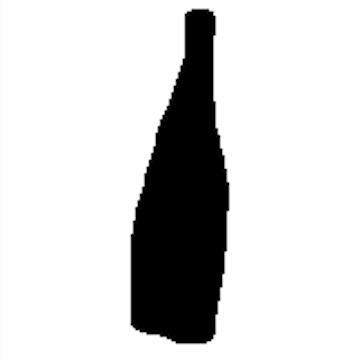} }}
    \quad
    \subfloat[Bone]{{\includegraphics[width=.25\textwidth]{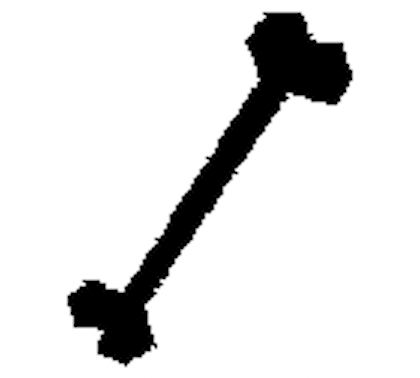} }}
    \quad
    \subfloat[Wrench]{{\includegraphics[width=.25\textwidth]{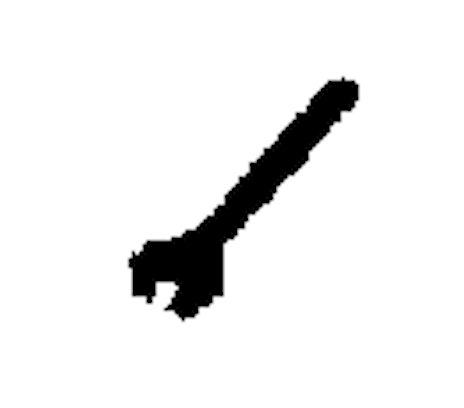} }}
    \caption{Examples from the Kimia Dataset.}
    \label{fig:bottle_bone_wrench}
\end{figure}

\begin{figure}[htp]
    \centering
    \subfloat[Bottle]{{\includegraphics[width=.25\textwidth]{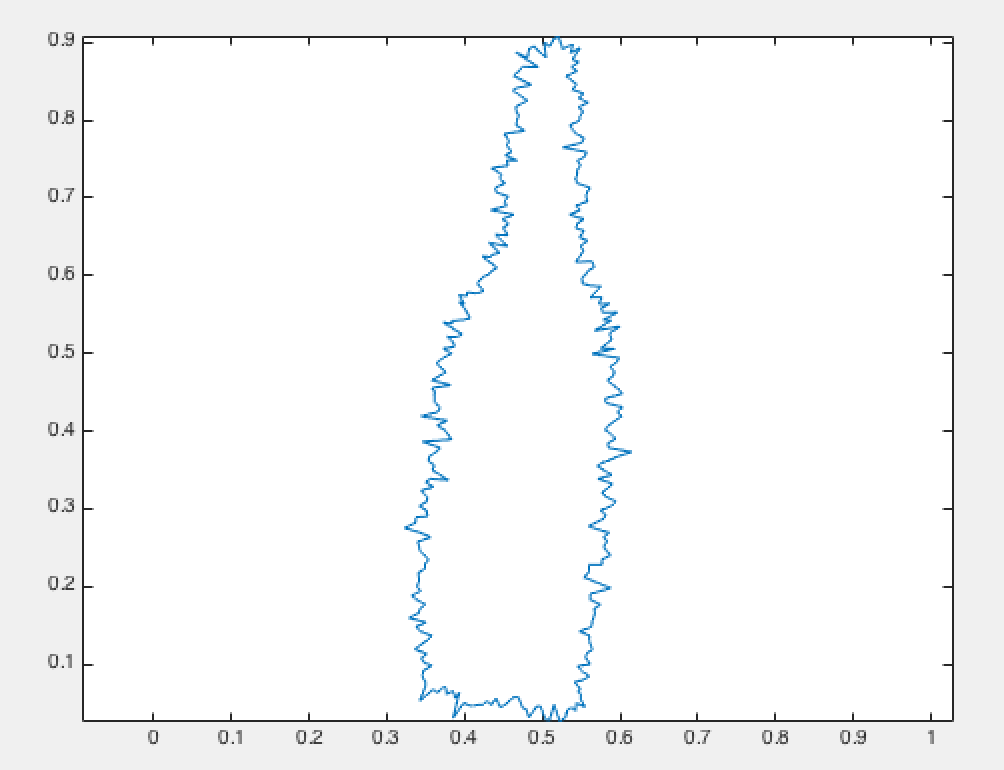} }}
    \quad
    \subfloat[Bone]{{\includegraphics[width=.25\textwidth]{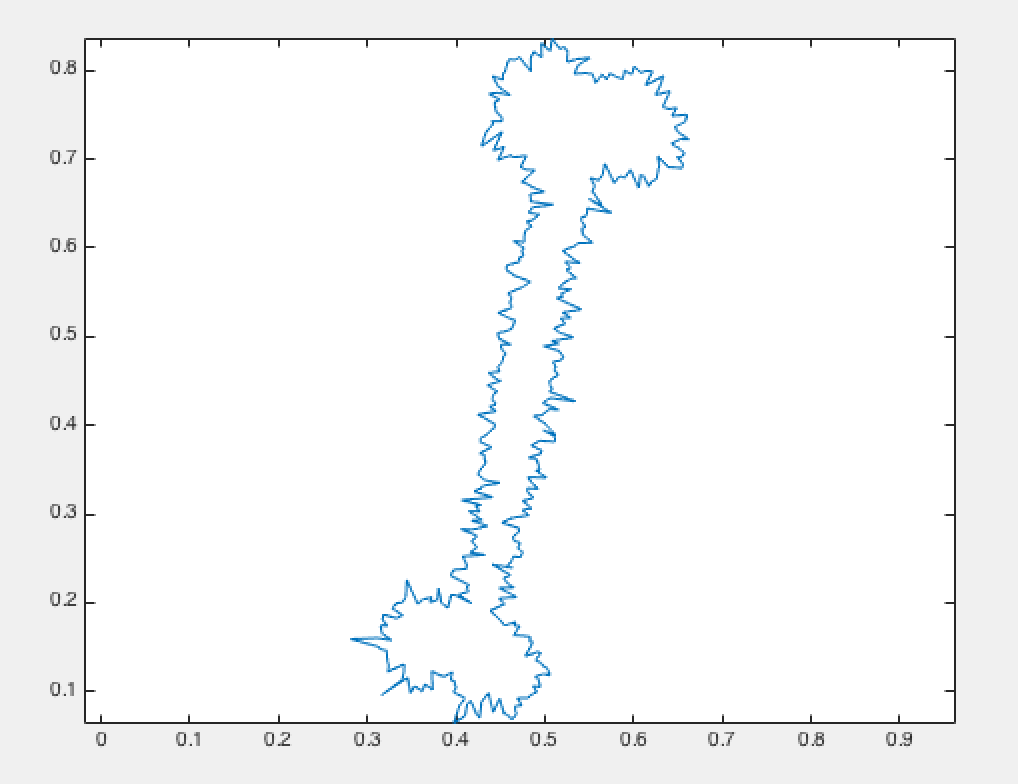} }}
    \quad
    \subfloat[Wrench]{{\includegraphics[width=.25\textwidth]{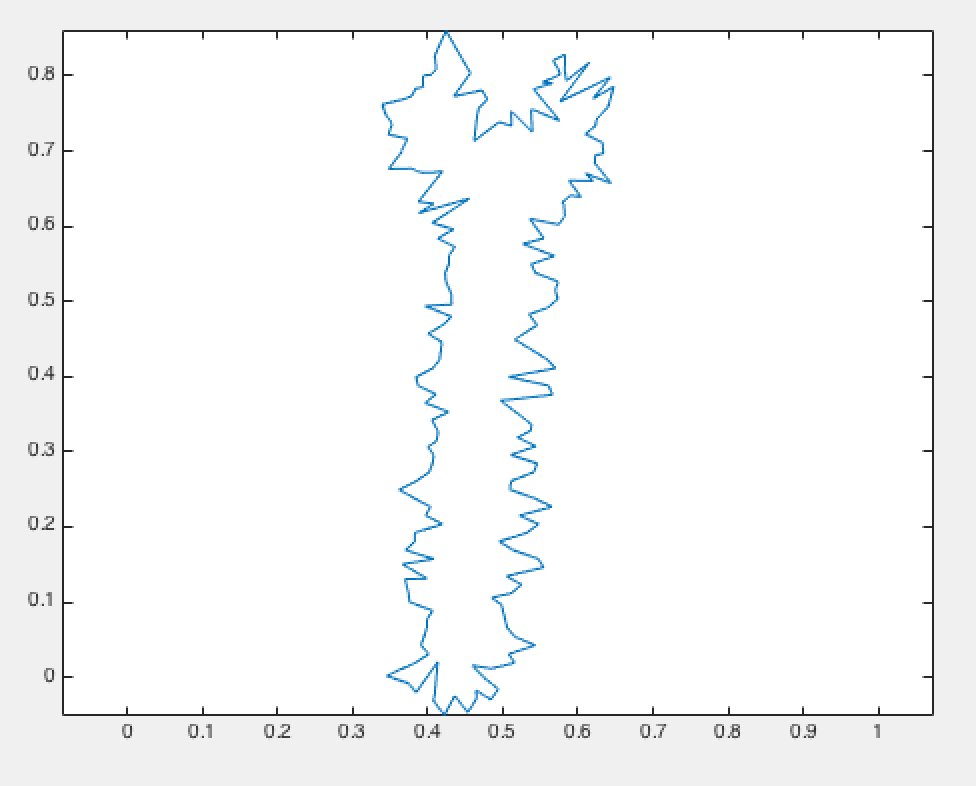} }}
    \caption{Noisy versions of examples from the Kimia Dataset.}
    \label{fig:noisy_bottle_bone_wrench}
\end{figure}

We then compute the pairwise distance matrix between all the curves (including the noiseless templates) based on a shape distance which was introduced in \cite{Glaunes2008} and later extended in the work of \cite{Kaltenmark2017}. This type of metric is based on the representation of shapes in a particular distribution space called currents, see \cite{Kaltenmark2017} for details. In our context, this metric offers several advantages: (i) the distance is completely geometrical in the sense that it is independent of the sampling of the curves and does not rely on predefined pointwise correspondences between vertices; (ii) it has an intrinsic smoothing effect that provides robustness to noise to a certain degree; (iii) it can be computed in closed form with minimal computational time which is critical given the large number of pairwise distances to evaluate. In this setting, we can view the resulting distance matrix as a perturbation of the ideal distances between the 9 template curves, which fits into the generic framework of our model. (Note that we leave aside the issue of checking the technical assumptions on the matrix $E$, which may be quite involved for this noise model and distance.)

We proceed to perform CMDS on this distance matrix. A scree plot investigation
shows that an appropriate embedding dimension here is $\hat{d} = 3$ (the top three eigenvalues are 2.20, 0.68, 0.06 with the fourth $\ll$ 0.01).
The resulting embedding configuration is shown in Figure \ref{fig:large currents}. 
This configuration exhibits nine fairly well-separated clusters roughly centered around the position of each of the noiseless template curves. Those, in turn, form 3 `super-clusters' consistent with the classes. Furthermore, the ellipsoidal shape of each cluster suggests that the configuration approximately follows a Gaussian distribution.

\begin{figure}[htp]
\centering
\includegraphics[scale=0.85]{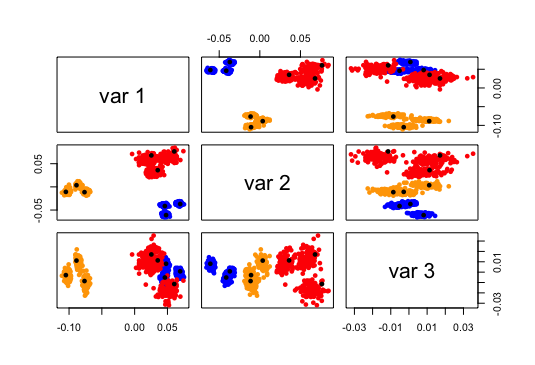}
\caption{Pairs plot of CMDS into $\mathbb{R}^3$ for the noisy curves. Colors correspond to the different classes (blue for bottle, red for bone, and orange for wrench). The position of the nine template curves in the configuration are highlighted with large black dots.}
    \label{fig:large currents}
\end{figure}

While these preliminary shape clustering results are obtained with a specific and simple distance on the space of curves, future work will investigate whether similar properties hold with different, more elaborate metrics and/or geometric noise models. The central limit theorem derived here could then constitute a useful theoretical tool to evaluate the discriminating power of shape clustering methods based on CMDS.

\section{Discussion}
\label{D}

In \citet{Athreya2016} and \citet{OMNI}, the authors prove that adjacency spectral embedding of the random dot product graph gives rise to a central limit theorem for the estimated latent positions. In this work we extend these results to the previously unexplored area of perturbation analysis for CMDS, addressing a gap in the literature as acknowledged in \citet{Fan} and \citet{Peterfreund&Gavish}. Notably, the three noise models we proposed in Section \ref{sec:M&E} each give rise to a central limit theorem; that is, for Euclidean distance matrix, the rows of the configuration matrix given by CMDS under noise will center around the corresponding rows of the true configuration matrix. Furthermore, our simulations on the synthetic data together with the shape clustering data all demonstrated the validity of our results. We have avoided any discussion of the model selection problem of choosing a suitable embedding dimension $\hat{d}$. Instead, we assume $d$ is known -- except in Section 4.2. There are many methods for choosing (spectral) embedding dimensions, see \cite{Zhu&Ghodsi, Jackson, Chatterjee}.

One Natural question can be raised is how to estimate the $\sigma$ in the noise model of interests. However, we would like to point out that for our embedding method and associated theoretical results, consistent estimation of $\sigma$ is not important. Indeed, the classical multidimensional scaling algorithm does not require estimating $\sigma$, but rather the dimension $d$ of the original data points (see the description of classical multidimensional scaling in section 1.1). Under all of our noise model, $\|\mathbf{E}\| \leq \sigma \sqrt{n}$ and provided that we choose $d$ such that $\lambda_d > n^{1/2 + \epsilon}$ for any $\epsilon > 0$, then our theoretical limit results apply. For concreteness, we can choose $\epsilon = 1/3$ and thus as long as we choose the embedding dimension $\hat{d}$ satisfying $\lambda_{\hat{d}}(B) \geq n^{2/3}$, then $\hat{d} \rightarrow d$ almost surely and our central limit theorem applies. 

Throught this paper, we assume that $d$ is fixed as $n \rightarrow \infty$. Therefore, given a central limit theorem for the embedding into $d$ dimension, one can derive a central limit theorem for the embedding into $d' < d$ dimension in a straightforward manner. More specifically, given a dissimilarity matrix $\hat{\Delta}^{(2)}$ and positive integers $d' \leq d$, the classical multidimensional scaling of $\hat{D}^{(2)}$ into $\mathbb{R}^{d'}$ is equivalent to the classical multidimensional scaling of $\hat{\Delta}^{(2)}$ into $\mathbb{R}^{d}$ and keeping the first $d' < d$ columns (see the description of classical multidimensional scaling in Section 1.1). Thus, our limit results can be rephrased to say that, letting $\hat{X}_n^{(d')}$ denote the classical multidimensional scaling of $\hat{D}^{(2)}$ into $\mathbb{R}^{d'}$ for $d' < d$, that there exists a sequence of $d' \times d'$ orthogonal matrix $W_n^{(d')}$ and a sequence of $d \times d'$ matrices with orthonormal columns $T_n$ such that
$$ \sqrt{n}\Bigl((\hat{X}_n^{(d')} W_n^{(d')})_{i} -  T_n( Z_n - \bar{Z}_n)_{i}\Bigr) $$
converges to a mixture of multivariate normal. For a given $n$, $T_n$ is a matrix corresponding the principal component projection of $Z_n$ into $\mathbb{R}^{d}$. We emphasize that $T_n$ is not necessarily unique (indeed, the eigenvalues of the covariance matrix for $Z_n$ are not necessarily distinct). 

We further note that the dependency on $d$ in our limit results is implicit in the covariance matrices. Naively speaking, we can say that the estimation accuracy is inversely proportional to $d$. This is most visible in the statement of Equation (1) (which is also a corollary of our results), since as $d$ increases $r$ also increases, note that $r \leq d + 2$. A more precise description is that the accuracy of our limit results depends on the covariance matrix $\Sigma$, which is a $d \times d$ matrix. Since the squared norm of a mean $0$ multivariate Gaussian is the trace of its covariance matrix, we see that as $d$ increases, the trace of $\Sigma$ does not have to increase with $d$. Indeed, the trace of $\Sigma$ depends purely on the distribution $F$ of the underlying data points; in the case where the data points are sampled from a multivariate normal with mean $0$ and identity matrix in $\mathbb{R}^{d}$, then as $d$ increases, the trace of $\Sigma$ also increases linearly.

Our presentation emphasizes the central limit theorem mainly because it is a succinct limit results. Nevertheless, the uniform or global error bounds can be established in a similar manner. More specifically, the central limit theorem for a fixed index $i$ is a consequence of applying the Lindeberg-Feller central limit theorem to Eq.(5) (which is a sum of independent mean $0$ random variables). If, instead of the Lindeberg-Feller central limit theorem, we apply a concentration inequality a la Hoeffding/Bernstein, then we can show that for any index $i$, $\|(\hat{X}_n W_n)_i - (Z_i - \bar{Z})\| \leq C n^{-1/2}$ with high probability. A union bound over the $n$ rows of $X_n$ then implies
$$\sup_{i \in [n]} \| (\hat{X}_n W_n)_{i} - (Z_i - \bar{Z}) \|  \leq C \sqrt{\frac{\log{n}}{n}}; \quad n^{-1} \sum_{i} \| (\hat{X}_n W_n)_{i} - (Z_i - \bar{Z}) \|  \leq C \sqrt{\frac{\log{n}}{n}}.$$

\begin{figure}[tp]
    \centering
    \subfloat[$n$=50]{{\includegraphics[width=.35\textwidth]{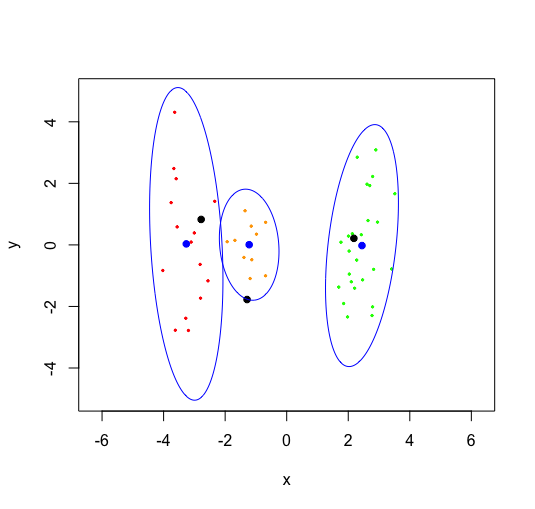} }}
    \quad
    \subfloat[$n$=100]{{\includegraphics[width=.35\textwidth]{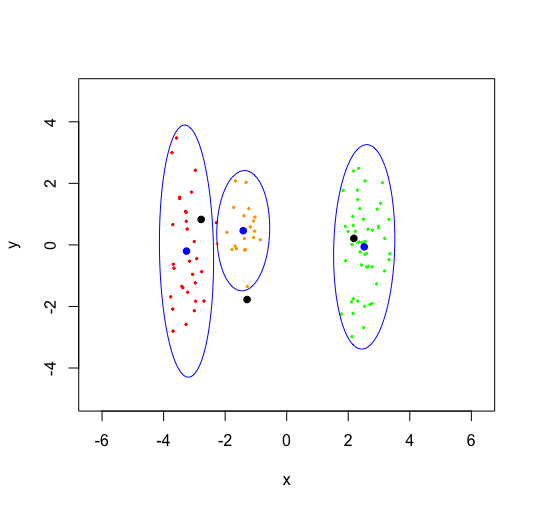} }}
    \quad
    \subfloat[$n$=500]{{\includegraphics[width=.35\textwidth]{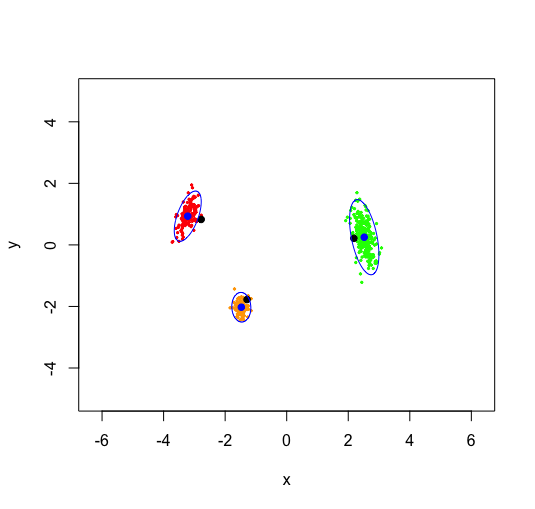} }}
    \quad
    \subfloat[$n$=1000]{{\includegraphics[width=.35\textwidth]{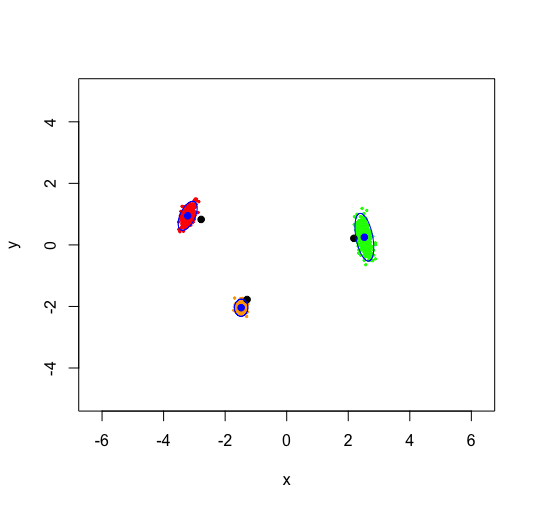} }}
    \caption{Simulation of CMDS with heteroscedastic noise $\widetilde{E}$. The black dots are the true positions for the three points. The blue dots are the empirical means and the blue ellipses are the 95\% level curve of the empirical covariance matrix. Note that $\widetilde{E}$ used in this simulation is of the same order for the off-diagonal blocks as that used in Figure \ref{fig:Simulation result}. NB: there is asymptotic bias.}
    \label{fig:Uncommon_var_E}
\end{figure}

A practically relevant and conceptually illustrative example comes from relaxing the assumption of common variance for the entries of the noise matrix $E$ in Section \ref{D+E}: the consistency result from Theorem \ref{main_theorem_D+E} no longer holds. To illustrate this point, we return to our three-point-mass simulation presented in Section \ref{SD} and modify our noise model as follows: Let $\widetilde{E}_{ij} \iid \textrm{Uniform}(-D_{ij}, +D_{ij})$ for $i < j$ and $\widetilde{E}_{ij}= \widetilde{E}_{ji}$. (The noise now depends on the entries of $D$, and $\Delta = D + \widetilde{E}$ no longer has negative entries.) The embedding of $\Delta$ into two dimensions gives class-conditional Gaussians; however, we have introduced bias into the embedding configuration. Figure \ref{fig:Uncommon_var_E} shows, for one realization, the embedding result. Note that the empirical mean and the theoretical positions do not coincide in simulation with large $n$, and theoretically even in the limit.
\begin{figure}[tp]
    \centering
    \subfloat[$n$=50]{{\includegraphics[width=.35\textwidth]{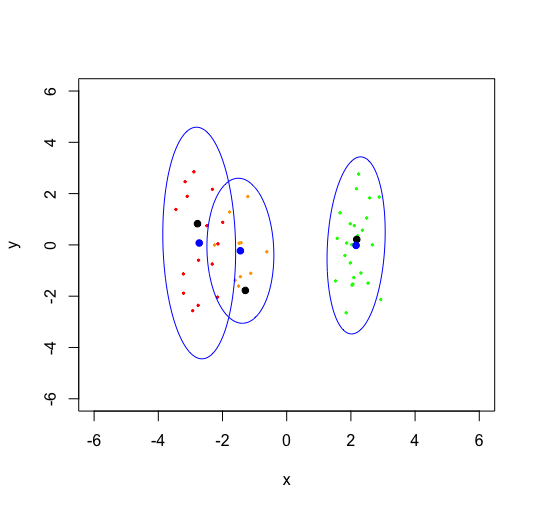} }}
    \quad
    \subfloat[$n$=100]{{\includegraphics[width=.35\textwidth]{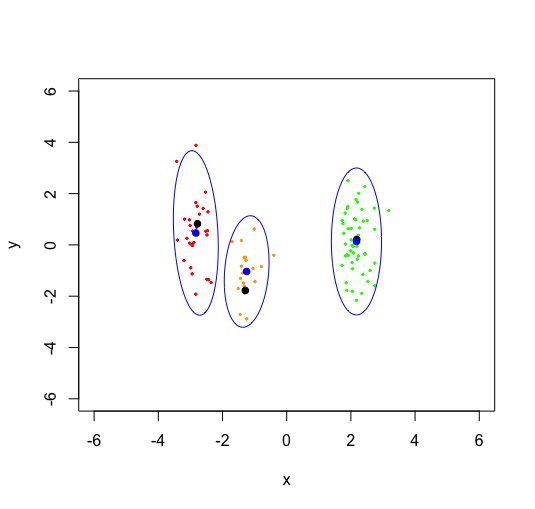} }}
    \quad
    \subfloat[$n$=500]{{\includegraphics[width=.35\textwidth]{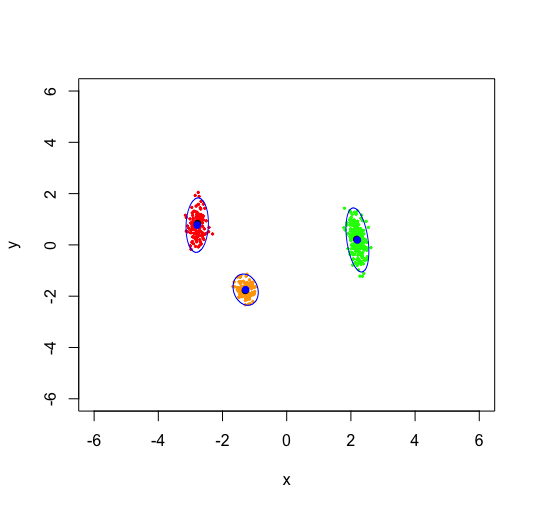} }}
    \quad
    \subfloat[$n$=1000]{{\includegraphics[width=.35\textwidth]{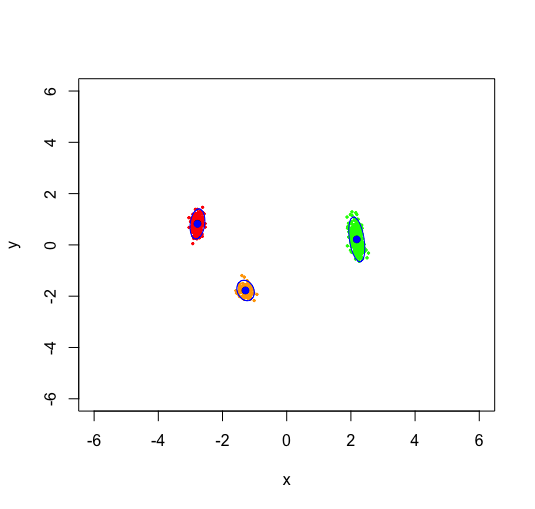} }}
    \caption{Simulation of MDS using raw stress criterion for $n$=50, 100, 500 and 1000 points.
                  The black dots are the true positions of $x_1$, $x_2$ and $x_3$, the blue dots are the empirical mean of the simulation and the blue ellipses are the 95\% level curve of the empirical covariance matrix.}
    \label{fig:RawStress}
\end{figure}

CMDS is just one of a wide variety of multidimensional scaling techniques. Minimizing the raw stress criterion is another commonly used MDS technique \citep{Leeuw-Heiser}, i.e., given a $n \times n$ observed dissimilarity matrix $\Delta$ and an embedding dimension $d$, one seeks to minimize the objective function $$\sigma_r = \sigma_r(X) = \sum\limits_{(i,j)} (\delta_{ij} - \|X_i - X_j\|)^2.$$ The minimization of $\sigma_r(X)$ is with respect to all configurations $X \in \mathbb{R}^{n \times d}$ and usually proceeds via an iterative algorithm which updates the configuration matrix $X$ until a stopping criterion is met. Keeping the simulation settings as in Section \ref{SD},  the resulting configuration is shown in Figure \ref{fig:RawStress}. This suggests that the CLT may hold for raw stress just as well as for CMDS. However, this claim is at best a conjecture at present as perturbation analysis of stress minimization algorithms is significantly more involved.

\appendix
\section*{Appendix: Proofs of stated results}
Throughout this Appendix, $\| A \|$ denotes the spectral norm of matrix $A$ and $\|A\|_F$ denotes its Frobenius norm. 
We will utilize the following observation repeatedly in our presentation.
\begin{observation}
\label{appthm1}
 Let $A$ and $B$ be matrices of appropriate dimensions. Then 
 $$\|A B\|_F = \|B^{\top} A^{\top}\|_{F} \leq \min\{ \|A\| \times \|B\|_F, \|B\| \times \|A\|_{F}\}.$$
\end{observation}

We remind our readers the following notations for the subsequent  presentation. Recall that $B = -\tfrac{1}{2} P D^2 P$ and $\hat{B} = -\frac{1}{2} P \Delta^2 P$ are the double centering of $D^2$ and $\Delta^2$, respectively. Note that if $D^2$ is a Euclidean distance matrix whose elements are $D_{ij} = \|Z_i -Z_j\|$, then $B = P Z Z^{\top} P$. Note then that $U_B S_B^{1/2} = P Z \tilde{W}_n$ for some $\tilde{W}_n$. Thus the $i$-th row of $U_B S_B^{1/2}$ is $\tilde{W}_n^{\top} (Z_i - \bar{Z})$ for some orthogonal $\tilde{W}_n$. Now let $W^{*}$ be the orthogonal matrix satisfying $W^{*} = \argmin_{W} \|U_B^{\top} \hat{U}_B - W\|$. Our main goal is to investigate the quantity $\hat{X} - U_B S_B^{1/2} W^* $. The following lemma provides a decomposition for $\hat{X} - U_B S_B^{1/2} W^*$ into a sum of several matrices. 
\begin{lemma}
\label{appthm2}
  Let $W^{*}$ be the orthogonal matrix satisfying $W^{*} = \argmin_{W} \|U_B^{\top} \hat{U}_B - W\|$. Then 
    \begin{align}
          \hat{X} - U_B S_B^{1/2} W^*  &=   ( \hat{B} - B) U_B S_B^{-1/2} W^{*}  \label{term1}\\
          & - (\hat{B} - B) U_B (S_{B}^{-1/2} W^{*} - W^{*} S_{\hat{B}}^{-1/2})     \label{term2}\\
          & - U_B U_B^{\top} (\hat{B} - B) U_B W^* S_{\hat{B}}^{-1/2}      \label{term3}\\
          & + (I - U_B U_B^{\top}) (\hat{B} - B) (U_{\hat{B}} - U_B W^{*}) S_{\hat{B}}^{-1/2}     \label{term4}\\
          & + U_B (U_B^{\top} U_{\hat{B}} - W^*) S_{\hat{B}}^{1/2}     \label{term5}\\
          & + U_B (W^{*} S_{\hat{B}}^{1/2} - S_B^{1/2} W^{*})     \label{term6}
    \end{align}
\end{lemma}
\begin{proof}
We have 
        \begin{align*}
          \hat{X} - U_B S_B^{1/2} W^*  
          & = U_{\hat{B}} S_{\hat{B}}^{1/2} - U_B W^* S_{\hat{B}}^{1/2} + U_B (W^* S_{\hat{B}}^{1/2} - S_{B}^{1/2} W^*) \\
          & = U_{\hat{B}} S_{\hat{B}}^{1/2} - U_B U_B^{\top} U_{\hat{B}} S_{\hat{B}}^{1/2} + U_B U_B^{\top} U_{\hat{B}} S_{\hat{B}}^{1/2} - U_B W^* S_{\hat{B}}^{1/2} + U_B (W^* S_{\hat{B}}^{1/2} - S_{B}^{1/2} W^*) \\
          & = (I - U_B U_B^{\top}) \hat{B} U_{\hat{B}} S_{\hat{B}}^{-1/2} + U_B (U_B^{\top} U_{\hat{B}} - W^{*}) S_{\hat{B}}^{1/2} + U_B (W^* S_{\hat{B}}^{1/2} - S_{B}^{1/2} W^*) \\
          &= (I - U_B U_B^{\top}) (\hat{B} - B) U_{\hat{B}} S_{\hat{B}}^{-1/2} + U_B (U_B^{\top} U_{\hat{B}} - W^*) S_{\hat{B}}^{1/2} + U_B (W^* S_{\hat{B}}^{1/2} - S_{B}^{1/2} W^*) 
        \end{align*}
        Note that we used the facts $ U_B U_B^\top B = B$ and $U_{\hat{B}} S_{\hat{B}}^{1/2} = \hat{B} U_{\hat{B}} S_{\hat{B}}^{-1/2}$ in the above equalities. The last two terms of the above display is Eq.~\eqref{term5} and Eq.~\eqref{term6} in the statement of the Lemma. We now consider the term 
        $(I - U_B U_B^{\top}) (\hat{B} - B) U_{\hat{B}} S_{\hat{B}}^{-1/2}$. 
        \begin{equation*}
        \begin{split}
        (I - U_B U_B^{\top}) (\hat{B} - B) U_{\hat{B}} S_{\hat{B}}^{-1/2}
        &
            =  (I - U_B U_B^{\top}) (\hat{B} - B) ( U_BW^* + \hat{U}_B - U_{B} W^{*}) S_{\hat{B}}^{-1/2}   \\
          & = (\hat{B} -B) U_B W^* S_{\hat{B}}^{-1/2} - U_B U_B^{\top} (\hat{B} - B) U_B W^* S_{\hat{B}}^{-1/2} + (I - U_B U_B^{\top}) (\hat{B} - B) (\hat{U}_B - U_{B} W^{*}) S_{\hat{B}}^{-1/2} 
          \\
          & = ( \hat{B} - B) U_B S_B^{-1/2} W^{*} 
           - (\hat{B} - B) U_B (S_{B}^{-1/2} W^{*} - W^{*} S_{\hat{B}}^{-1/2}) 
           - U_B U_B^{\top} (\hat{B} - B) U_B W^* S_{\hat{B}}^{-1/2} \\
           & + (I - U_B U_B^{\top}) (\hat{B} - B) (\hat{U}_B - U_{B} W^{*})S_{\hat{B}}^{-1/2}
           \end{split}
        \end{equation*}
        The four terms in the above display correspond to the matrices in Eq.~\eqref{term1} through Eq.~\eqref{term4}.
\end{proof}

 Note that from Lemma 5, we have $\hat{X} {W^*}^{\top} \tilde{W}_n - U_B S_B^{1/2} \tilde{W}_n = (\hat{B} - B) U_B S_B^{-1/2} \tilde{W}_n$ + remaining terms in Eq.~\eqref{term2} through Eq.~\eqref{term6}. The essential term is $(\hat{B} - B) U_B S_B^{-1/2} \tilde{W}_n$ and we analyzed the rows of this matrix in Lemma~\ref{appthm3} below where we show that they converge to multivariate normals. We then show in Lemma~\ref{appthm4} below shows that the rows of the remaining matrices in Eq.~(\ref{term2}) through Eq.~(\ref{term6}), when scaled by $\sqrt{n}$, converge to $0$ in probability. Combining these results yield the proof of Theorem 2. Indeed, the term $\hat{X} {W^*}^{\top} \tilde{W}_n$ can be denoted by $\hat{X} W_n$ for some orthogonal matrix ${W}_n= {W^*}^{\top} \tilde{W}_n$ identical to that in the statement of Theorem 1,2, and 3, while the rows of $U_B S_B^{1/2} \tilde{W}_n$ is, as we observed earlier, simply $(Z_i - \bar{Z})$. 

\begin{lemma}
\label{appthm3}
  Let the rows of $X$: $X_k \stackrel{i.i.d}{\sim}$ F for some sub-Gaussian distribution F. Then there exists a sequence of $d \times d$ orthogonal matrices $\tilde{W}_n$, such that for any fixed index $i$, we have
  $$ \sqrt{n} \tilde{W}_n^{\top} [(\hat{B} - B) U_B S_B^{-1/2}]_{i} \overset{\mathcal{L}}{\to} \mathcal{N}(0, \Sigma(x_i))$$
  where $\Sigma (x_i)= {\Xi}^{-1} \widetilde{\Sigma}(x_i) {\Xi}^{-1}$, 
  $\Xi = \mathbb{E}[X_k{X_k}^\top] \in \mathbb{R}^{d \times d}$, ${\mu} = \mathbb{E}[X_k] \in \R^d.$ and 
  $$\widetilde{\Sigma} (x_i) = \mathbb{E}_{X_k}[(\sigma^2 ||x_i - X_k||^2 + \mathbb{E}[E_{ij}^3] \|x_i - X_k \|+ \frac{1}{4} \mathbb{E}[E_{ij}^4] - \frac{\sigma^4}{4}) (X_k- {\mu}) (X_k - {\mu})^\top] \in \mathbb{R}^{d \times d}$$  is a covariance matrix depending on $x_i$. Here, for ease of notation, we denote by $(A)_i$ or $[A]_i$ the $i$-th row of matrix $A$. 
\end{lemma}
\begin{proof}
 Recall, since $X = U_B S_B^{1/2} W_n $, we can write, as $n \rightarrow \infty$
    \begin{align*}
        \sqrt{n} \tilde{W}_n^\top [(\hat{B} - B) U_B S_B^{-1/2}]_{i}
        & = \sqrt{n} \tilde{W}_n^\top [(\hat{B} - B)X \tilde{W}_n^\top {S_B}^{-1}]_{i} \\
        & = \sqrt{n} \tilde{W}_n^\top {S_B}^{-1} \tilde{W}_n [(\hat{B} - B)X]_{i}\\
        & = - \sqrt{n} \tilde{W}_n^\top {S_B}^{-1} \tilde{W}_n [P(D \circ E  + \frac{E^2}{2} )PX]_{i}\\
        & = - \sqrt{n} \tilde{W}_n^\top {S_B}^{-1} \tilde{W}_ [ (I - \frac{\boldsymbol{1}\boldsymbol{1}^\top}{n}) (D \circ E  + \frac{E^2}{2}) (I - \frac{\boldsymbol{1}\boldsymbol{1}^\top}{n}) X ]_{i}\\
        & = - \sqrt{n} \tilde{W}_n^\top {S_B}^{-1} \tilde{W}_n [  (I - \frac{\boldsymbol{1}\boldsymbol{1}^\top}{n}) (D \circ E  + \frac{E^2}{2}) (X - \bar{X}) ]_{i}\\
        & = - \sqrt{n} \tilde{W}_n^\top {S_B}^{-1} \tilde{W}_n [ (I - \frac{\boldsymbol{1}\boldsymbol{1}^\top}{n}) ( D \circ E  + \frac{E^2}{2} - \frac{\sigma^2 \boldsymbol{1}\boldsymbol{1}^\top}{2} + \frac{\sigma^2 \boldsymbol{1}\boldsymbol{1}^\top}{2}) (X - \bar{X})]_{i}. \\
        & = - \sqrt{n} \tilde{W}_n^\top {S_B}^{-1} \tilde{W}_n [ (I - \frac{\boldsymbol{1}\boldsymbol{1}^\top}{n}) (D \circ E + \frac{E^2 - \sigma^2 \boldsymbol{1}\boldsymbol{1}^\top}{2})(X - \bar{X})]_{i} \\
    \end{align*}
Note the last equality holds since  $(I-\frac{\boldsymbol{1}\boldsymbol{1}^\top}{n})\frac{\sigma^2\boldsymbol{1}\boldsymbol{1}^\top}{2}(X-\bar{X}) = 0$, 
hence
$$\sqrt{n} \tilde{W}_n^\top [(\hat{B} - B) U_B S_B^{-1/2}]_{i}  = - \sqrt{n} \tilde{W}_n^\top {S_B}^{-1} \tilde{W}_n[ (D \circ E + \frac{E^2 - \sigma^2 \boldsymbol{1}\boldsymbol{1}^\top}{2}) (X- \bar{X})]_{i} $$
as 
$(X - \boldsymbol{1} \mu^{\top} + \boldsymbol{1} \mu^{\top} - \bar{X} )$ has mean 0 and $\frac{\boldsymbol{1}\boldsymbol{1}^\top}{n}(D \circ E + \frac{E^2 - \sigma^2 \boldsymbol{1}\boldsymbol{1}^\top}{2})(X - \boldsymbol{1} \mu^{\top} + \boldsymbol{1} \mu^{\top} - \bar{X}) \xrightarrow{{n \rightarrow \infty}} 0$. We therefore have
    \begin{align*}
        \sqrt{n} \tilde{W}_n^\top [(\hat{B} - B) U_B S_B^{-1/2}]_{i} 
         =& - n \tilde{W}_n^\top {S_B}^{-1} \tilde{W}_n [\frac{1}{\sqrt{n}} ( \sum\limits_{j \neq i}^{n} [(D \circ E + \frac{E^2 - \sigma^2 \boldsymbol{1}\boldsymbol{1}^\top}{2})_{ij} (X - \boldsymbol{1} \mu^{\top})_{j}])\\ 
        & - \frac{1}{\sqrt{n}} (D \circ E + \frac{E^2 - \sigma^2 \boldsymbol{1}\boldsymbol{1}^\top}{2})_{ii} (X- \boldsymbol{1} \mu^{\top})_{i} ].
    \end{align*} 
 
 Note  $\frac{1}{\sqrt{n}} (D \circ E + \frac{E^2 - \sigma^2 \boldsymbol{1}\boldsymbol{1}^\top}{2})_{ii} (X-\boldsymbol{1} \mu^{\top})_{i} \xrightarrow{{n \rightarrow \infty}} 0 $, hence when $n \rightarrow \infty$, the above expression yields:
 
 \begin{equation} \label{eq:6}
  - n \tilde{W}_n^\top {S_B}^{-1} \tilde{W}_n [\frac{1}{\sqrt{n}} ( \sum\limits_{j \neq i}^{n} [(D_{ij} \cdot E_{ij} + \frac{E_{ij}^2 - \sigma^2 \boldsymbol{1}\boldsymbol{1}^\top}{2}) (X_{j} - {\mu}^\top) ] )]
  \end{equation}
Condition on $X_i = x_i$, (\ref{eq:6}) is then the sum of $n-1$ independent mean $0$ random variables, each with covariance matrix given by:
\begin{align*}
\mathrm{Cov}[(E_{ij} \|x_i - X_j\| + \frac{E_{ij}^2 - \sigma^2}{2}) (X_j - {\mu}^\top)]
& = \sum\limits_{j \neq i}^{n} \mathrm{Var}(E_{ij} \|x_i - X_j\| + \frac{E_{ij}^2 - \sigma^2}{2}) (X_j - {\mu}^\top) (X_j - {\mu}^\top)^\top
\end{align*}

We now consider $\mathrm{Var}(E_{ij} \|x_i - X_j\| + (E_{ij}^2 - \sigma^2)/2)$. Since $\mathbb{E}[E_{ij}] = 0$ and $\mathbb{E}[E_{ij}^2] = \sigma^2$, we have
$$
\mathrm{Var}\Bigl(E_{ij} \|x_i - X_j\| + (E_{ij}^2 - \sigma^2)/2\Bigr) = \mathbb{E}\Bigl[ E_{ij}^2 \|x_i - X_j\| + E_{ij} \|x_i - X_j\| (E_{ij}^2 - \sigma^2) + \frac{(E_{ij}^2 - \sigma^2)^2}{4}\Bigr]$$
where the expectation is taken with respect to $E_{ij}$ and conditional on $X_j$.
Hence
$$
\widetilde{\Sigma} (x_i) = \mathbb{E}_{X_k}\Bigl[(\sigma^2 ||x_i - X_k||^2 + \mathbb{E}[E_{ij}^3] \| x_i - X_k\| + \tfrac{1}{4} \mathbb{E}[E_{ij}^4] - \tfrac{\sigma^4}{4}) (X_k- {\mu}) (X_k - {\mu})^\top\Bigr].
$$
Finally, by the strong law of large numbers, we have
  $$\frac{\tilde{W}_n^\top S_B \tilde{W}_n}{n} = \frac{1}{n} X^\top X {\to} \Xi \in \mathbb{R}^{d \times d}$$ almost surely. Hence $(n \tilde{W}_n^{\top} S_B^{-1} \tilde{W}_n) {\to} {\Xi}^{-1}$ almost surely. Slutsky's theorem then yields $$ \sqrt{n} \tilde{W}_n^\top [(\hat{B} - B) U_B S_B^{-1/2}]_{i} \overset{\mathcal{L}}{\to} \mathcal{N}(0, {\Xi}^{-1} \widetilde{\Sigma}(x_i) {\Xi}^{-1})$$
  as desired.  
\end{proof}  

We now look at the matrices in Eq.~(\ref{term2}) through 
Eq.~(\ref{term6}). The following lemma show that any row of these matrices, when scaled by $\sqrt{n}$, will converge to $0$ in probability. 
\begin{lemma}
\label{appthm4}
  We have, simultaneously
  \begin{gather} \label{eq:1}
    \sqrt{n} [(\hat{B} - B) U_B (W^{*} S_{\hat{B}}^{-1/2} - S_B^{-1/2} W^{*})]_{h} \overset{P}{\to} 0 \\
 \label{eq:2}
    \sqrt{n} [ U_B U_B^{\top} (\hat{B} - B) U_B W^{*} S_{\hat{B}}^{-1/2}]_{h} \overset{P}{\to} 0 \\
 \label{eq:3}
    \sqrt{n} [ (I -U_B U_B^{\top}) (\hat{B} - B) (\hat{U}_B - U_B W^*) S_{\hat{B}}^{-1/2}]_{h} \overset{P}{\to} 0 \\
 \label{eq:4}
 \sqrt{n}[U_B (U_B^{\top} U_{\hat{B}} - W^*) S_{\hat{B}}^{1/2}]_{h}  \overset{P}{\to} 0. \\
\label{eq:5}
\sqrt{n}[U_B (W^{*} S_{\hat{B}}^{1/2} - S_B^{1/2} W^{*})]_{h} \overset{P}{\to} 0.
  \end{gather}
\end{lemma}
The rest of this Appendix is devoted toward proving Lemma \ref{appthm4}, for which we need the following technical lemmas controlling the spectral norm of $\|\hat{B} - B\|$ and $\|U_B^{\top} \hat{U}_B - W^{*}\|$ (recall that $W^*$ is the closest orthogonal matrix, in Frobenius norm, to $U_B^{\top} \hat{U}_B$.) 
We start with a bound for the spectral norm of $B - \hat{B}$. 

\begin{proposition}
\label{appthm6}
  $\|B - \hat{B}\| = \mathcal{O}(\sqrt{n \log n})$ with high probability.
\end{proposition}
\begin{proof}
 We have
  \begin{align*}
    \|B - \hat{B}\| & = \| -\frac{1}{2} P D^2 P +  \frac{1}{2} P (D+E)^2 P\|\\
    & = \|P D \circ E P + \frac{1}{2} P E^2 P\| \textrm{ (where $\circ$ is the Hadamard product)}\\      
    & \leq \|D \circ E\| + \frac{1}{2} \| E^2 - \mathbb{E} [E^2] \| \textrm{ (since $\|P\| = 1.)$ }\\
    & = \mathcal{O}(\sqrt{n}) + \mathcal{O}(\sqrt{n \log n})
  \end{align*}
  Note that here we used $\mathbb{E}[D \circ E] = 0$ and $\mathbb{E}[\frac{1}{2}P E^2 P ] = 0$. Each entries of $D \circ E$ is of sub-Gaussian distribution with mean $0$ and each entries of $E^2 - \mathbb{E} [E^2]$ is of sub-exponential distribution with mean $0$. An application of Theorem 4.4.5 in \cite{HDP} and Matrix Bernstein for the sub-exponential case in \cite{tropp2012user} gives the desired result.
\end{proof}

\begin{lemma}
\label{appthm7}
  Let $X_1, \ldots, X_n, Y \stackrel{i.i.d}{\sim} F$ for some sub-Gaussian distribution $F$, where $X_i$ is the $i$th row of the configuration matrix $X$ of $B$ viewed as a column vector. Let $\Xi = \mathbb{E}[X_1 {X_1}^\top]$ be of rank $d$, then $\lambda_i(B) = \Omega(n)$ almost surely.
\end{lemma}
\begin{proof}
  For any matrix $H$, the nonzero eigenvalues of $H^\top H$ are the same as those $HH^\top $, so $\lambda_i(XX^\top) = \lambda_i(X^\top X)$. In what follows, we remind the reader that $X$ is a matrix whose rows are the transposes of the column vectors $X_i$, and $Y$ is a d-dimensional vector that is independent from and has the same distribution as that of the $X_i$.
  We observe that $(X^\top X - n \mathbb{E}[YY^\top] )_{ij} = \sum\limits_{k=1}^{n} (X_{ki}X_{kj} - \mathbb{E}[Y_iY_j])$ is a sum of $n$ independent mean-zero sub-Gaussian random variables.
By a general Hoeffding's inequality for sub-gaussian random variables \citep{HDP}, for all $i, j \in [d]$, 
$$\mathbb{P}[|(X^\top X - n\mathbb{E}[YY^\top])_{ij}|  \geq t]  \leq 2 \exp\{ \frac{-ct^2}{nM} \},$$  where $M = \max\limits_{k} 
\|(X_{ki}X_{kj} - \mathbb{E}[Y_iY_j] )\|_{\varphi_2}^2$. Therefore,
$$ \mathbb{P}[|(X^\top X - n\mathbb{E}[YY^\top])_{ij}| \geq C\sqrt{n\log n}]  \leq 2 n^{\frac{-2C^2}{M^2}}.$$
A union bound over all $i, j \in [d]$ implies that $\|X^\top X - n\mathbb{E}[YY^\top ]\|_{F}^2 \leq C^2 d^2 n \log n$ with probability at least 
$1 - 2 n^{-2C^2/M^2}$, i.e. $\|X^\top X - n\mathbb{E}[YY^\top ] \|_{F} \leq C d \sqrt{n \log n }$ with high probability for any $C > \frac{M}{\sqrt{2}}.$
By the Hoffman-Wielandt inequality, $|\lambda_i(XX^\top) - n \lambda_i(\mathbb{E}[YY^\top ])| \leq C d \sqrt{n \log n}$, and by reverse triangle inequality, we obtain $$\lambda_i(XX^\top ) \geq \lambda_d(XX^\top) \geq | n \lambda_d(\Xi) | - C d \sqrt{n \log n} = \Omega(n)$$ holds almost surely.
\end{proof}

\begin{proposition}
\label{appthm8}
  Let $W_{1} \Sigma {W_2}^{T}$ be the singular value decomposition of $U_B^{\top} U_{\hat{B}}$, then with high probability, $\|U_B^{\top} U_{\hat{B}} - {W_1} {W_2}^\top\| = \mathcal{O}(n^{-1} \log n)$.
\end{proposition}
\begin{proof}
  Let $\sigma_1, \sigma_2, \ldots, \sigma_d$ be the singular values of $U_B^{\top} U_{\hat{B}}$ (the diagonal entries of $\Sigma$). Then  $\sigma_i = \cos(\theta_i)$ where $\theta_i$'s are the principal angles between the subspace spanned by $U_B$ and $U_{\hat{B}}$. The Davis-Kahan $\sin(\Theta)$  theorem \citep{D-K} gives
   
  $$\|U_{\hat{B}} U_{\hat{B}}^{\top} - U_B U_B^{\top}\| = \max\limits_{i} |\sin(\theta_i)| \leq \frac{C \|B - \hat{B}\|}{\lambda_d (B)} = \mathcal{O}(\sqrt {\frac{\log n}{n}})$$ 
  
  for sufficiently large $n$. Note in the last equality we used the previous two lemmas.
  Thus,
  \begin{align*}
    || U_B^{\top} U_{\hat{B}} - W_1 {W_2}^\top||_F
    & = || \Sigma - I ||_F = \sqrt{ \sum\limits_{i=1}^{d} (1- \sigma_i)^2 } \leq \sum\limits_{i=1}^{d} (1- \sigma_i) \leq \sum\limits_{i=1}^{d} (1- {\sigma_i}^2) \\
    &  = \sum\limits_{i=1}^{d} {\sin(\theta_i)}^2 \leq d || U_{\hat{B}} U_{\hat{B}}^{\top} - U_B U_B^{\top}||^{2} = \mathcal{O}(\frac{\log n}{n})
   \end{align*}
\end{proof}

Recall that a random vector $X$ is sub-exponential if $\mathbb{P}[ | X | > t] \leq 2 e^{ -\frac{t}{ {K}} }$ for some constant $K$ and for all $t \geq 0$. Associated with a sub-exponential random variable there is a Orlicz norm defined as $ \| X \|_{\psi_1} = \inf \{ t >0 : \mathbb{E} \exp(\frac{|X| }{t}) \leq 2 \}$. Furthermore, a random variable $X$ is sub-Gaussian if and only if $X^2$ is sub-exponential, and $ \| X^2\|_{\psi_1} = \|X\|_{\psi_2}^2 $. We now have the following lemma which allows us to juxtapose the ordering in the matrix product $W^{*} \hat{S}_B$ and $S_B W^{*}$ (and similarly $W^{*} \hat{S}_B^{1/2}$ and $S_B^{1/2} W^{*}$.) This juxtaposition is essential in showing Eq.~\eqref{eq:1} and Eq.~\eqref{eq:5} in Lemma~\ref{appthm4}. 

\begin{lemma}
\label{appthm9}
  Let $W^{*} = W_1{W_2}^\top$. Then with high probability,
  $$\|W^{*}S_{\hat{B}} - S_{B}W^{*}\|_{F} = \mathcal{O}(\log n); \quad \text{and} \quad \|W^{*} S_{\hat{B}}^{1/2} - S_{B}^{1/2} W^{*}\|_{F} = \mathcal{O}(n^{-\frac{1}{2}} \log n).$$
\end{lemma}
\begin{proof}
  Let $R = U_{\hat{B}} - U_BU_B^{\top} U_{\hat{B}}.$ Note $R$ is the residual after projecting $U_{\hat{B}}$ orthogonally onto the column space of  $U_{B}$, and thus $\|U_{\hat{B}} - U_B U_B^{\top} U_{\hat{B}}\|_{F} \leq \min\limits_{W} \|U_{\hat{B}} - U_B W \|_{F}$ where the minimization is over all orthogonal matrices $W$. 
  By a variant of the Davis-Kahan $\sin \Theta$ theorem \citep{vD-K}, we have
   $$\min\limits_{W}\|U_BW - U_{\hat{B}} \|_{F} \leq \frac{C \sqrt{d} \|B - \hat{B}\|}{\lambda_d(B)} ,$$ and hence $\|R\|_F \leq \mathcal{O}(\sqrt{\frac{\log n}{n}}).$
   Now consider 
   \begin{align*}
    W^{*}S_{\hat{B}}
    & = (W^{*} - U_B^{\top} U_{\hat{B}}) S_{\hat{B}} + U_B^{\top} U_{\hat{B}} S_{\hat{B}} \\
    & = (W^{*} - U_B^{\top} U_{\hat{B}}) S_{\hat{B}} + U_B^{\top} \hat{B} U_{\hat{B}}  \\
    & = (W^{*} - U_B^{\top} U_{\hat{B}}) S_{\hat{B}} + U_B^{\top} (\hat{B} - B) U_{\hat{B}} + U_B^{\top} B U_{\hat{B}} \\
    & = (W^{*} - U_B^{\top} U_{\hat{B}}) S_{\hat{B}} +  U_B^{\top} (\hat{B} - B) R + U_B^{\top} (\hat{B} - B) U_B U_B^{\top} U_{\hat{B}} + S_B U_B^{\top} U_{\hat{B}}.
   \end{align*}
   Note here we use the fact $U_{\hat{B}} S_{\hat{B}} = \hat{B} U_{\hat{B}}.$
   Now write $$S_B U_B^{\top} U_{\hat{B}} = S_B(U_B^{\top} U_{\hat{B}} - W^{*}) + S_B W^{*},$$
   then we have $$W^{*}S_{\hat{B}} - S_B W^{*}  = (W^{*} - U_B^{\top} U_{\hat{B}}) S_{\hat{B}} + U_B^{\top} (\hat{B} -B) R + U_B^{\top} (\hat{B} -B) U_B U_B^{\top} U_{\hat{B}} + S_B (U_B^{\top} U_{\hat{B}} - W^{*}).$$
   This gives 
   \begin{align*}
     \begin{array}{rl}
      \|W^{*}S_{\hat{B}} - S_B W^{*}\|_{F} &
      \leq \| (U_B^{\top} U_{\hat{B}} - W^{*}) (S_{\hat{B}} + S_B) \|_{F} + \| U_B^{\top} (\hat{B} - B) R \|_{F} + \| U_B^{\top} (\hat{B} - B) U_B U_B^{\top} U_{\hat{B}}\|_{F} \\
      & \leq  \| (U_B^{\top} U_{\hat{B}} - W^{*})\|_{F}  (\|S_{\hat{B}}\| + \|S_B\|) + \| U_B^{\top} (\hat{B} - B) R \|_{F} + \| U_B^{\top} (\hat{B} - B) U_B U_B^{\top} U_{\hat{B}}\|_{F} \\
      & \leq  \|W_1 W_2^{\top} - U_B^{\top} {U_{\hat{B}}}\|_{F} (\mathcal{O}(n) + \mathcal{O}(n)) + \| U_B^{\top} (\hat{B} - B) R \|_{F} + \| U_B^{\top} (\hat{B} - B) U_B \|_{F} \\
      & \leq \mathcal{O}(n^{-1}) (\mathcal{O}(n) + \mathcal{O}(n)) + \mathcal{O}(\log n)  + \| U_B^{\top} (\hat{B} - B) U_B \|_{F} \\
      & = \mathcal{O}( \log n) + \| U_B^{\top} (\hat{B} - B) U_B \|_{F}.
     \end{array}
   \end{align*}
   Now consider the term $ U_B^{\top} (\hat{B} - B) U_B \in \mathbb{R}^{d \times d}$. If we denote $U_i$ be the $i$th column of $U_B$, then for each $i, j$th entry, we have 
   \begin{align*}
    ( U_B^{\top} (\hat{B} - B) U_B )_{ij} = U_i^{\top} (\hat{B} -B) U_j = \frac{1}{2} V_{i}^{\top} (\Delta^2 - D^2) V_j
   \end{align*}
   where $V = P U_B$.
    Furthermore, we have
    \begin{equation} \label{eq:sum}
    V_{i}^{\top} (\Delta^2 - D^2) V_j = \sum\limits_{k,l} V_{ik} ({\Delta_{kl}}^2 - {D_{kl}}^2) V_{jl}.
    \end{equation}
 Recall, since $X_k$'s are sub-Gaussian, thus equation (\ref{eq:sum}) is a sum of mean zero sub-exponential random variables. By Bernstein's inequality \citep{HDP}, we have
 $$ \mathbb{P}[ |\sum\limits_{k,l} ({\Delta_{kl}}^2 - {D_{kl}}^2) V_{ik} V_{jl} | > t] \leq  2 \exp \Bigl\{-C \min (\frac{t^2}{M^2 \sum_{k,l} {V_{ik}}^2{V_{kl}}^2}, \frac{t}{M \max_{k,l} ( V_{ik} V_{jl})}) \Bigr \} $$  
 where $M := \max_{k,l} \| {\Delta_{kl}}^2 - {D_{kl}}^2 \|_{\psi_1} $.
 Since $\sum_{k} {V_{ik}}^2 \leq 1 \forall i $, we have that each entry of $ U_B^{\top} (\hat{B} - B) U_B  \in \mathbb{R}^{d \times d}$ is $\mathcal{O}(\log n)$, and 
 \begin{equation}
 \label{important_bound}
 \| U_B^{\top} (\hat{B} - B) U_B \|_{F} = \mathcal{O} (\log n).
 \end{equation}
 This then gives $\|W^{*}S_{\hat{B}} - S_{B}W^{*}\|_{F} = \mathcal{O}(\log n)$, with high probability.\\
    Finally, consider $ \|W^{*} S_{\hat{B}}^{1/2} - S_{B}^{1/2} W^{*}\|_{F}.$ The $i, j$th entry of $W^{*} S_{\hat{B}}^{1/2} - S_{B}^{1/2} W^{*}$ is 
    \begin{align*}
    {{W^{*}}_{ij}} ( {\lambda_{j}}^{1/2}(\hat{B}) -  {\lambda_{i}}^{1/2}({B})) 
    = {{W^{*}}_{ij}} \frac{{\lambda_{j}}(\hat{B}) -  {\lambda_{i}}({B})} {{\lambda_{j}}^{1/2}(\hat{B}) +  {\lambda_{i}}^{1/2}({B})} 
     \leq {{W^{*}}_{ij}} \frac{{\lambda_{j}}(\hat{B}) -  {\lambda_{i}}({B})}{\Omega(\sqrt{n})} 
     = \mathcal{O}(n^{-\frac{1}{2}} \log  n),
    \end{align*}
    as desired (note in the last inequality, we used the first part of this Lemma.
\end{proof}
We now proceed to prove Lemma \ref{appthm4}.
\begin{proof}[Proof of Lemma~\ref{appthm4}]
  To show Eq.~\eqref{eq:1}, we have
  \begin{align*}
    \sqrt{n} \| (\hat{B} -B) U_B (W^{*} S_{\hat{B}}^{-1/2} - S_B^{-1/2} W^{*}) \|_F 
    & \leq \sqrt{n} \| (\hat{B} -B) U_B \| \times \| W^{*} S_{\hat{B}}^{-1/2} - S_B^{-1/2} W^{*} \|_F \\
    &  \leq  \sqrt{n} \| (\hat{B} -B) \| \times \| W^{*} S_{\hat{B}}^{-1/2} - S_B^{-1/2} W^{*} \|_F \\
    & = \sqrt{n} \mathcal{O}(\sqrt{n \log n}) \mathcal{O}(n^{-\frac{3}{2}} \log n) = \frac{C \log n \sqrt{\log n}}{\sqrt{n}}
  \end{align*}
which converges to $0$ as $n \rightarrow \infty$.

Let us now consider Eq.~\eqref{eq:2}. Recall that $X = U_B S_B^{1/2} W$ for some orthogonal matrix W, and since $X_i$'s are sub-Gaussian, $\|X_i\|$ is bounded by some constant $C$ with high probability, i.e., $\|X_i\| = \sqrt{\sum\limits_{j=1}^d \sigma_j {{U_B}_{ij}}^2} \leq C$ with high probability, where $\sigma_i$'s are the diagonal entries of $S_B^{1/2}$. Note that 
$\sigma_i = \Omega(n) \geq C^{'}n$ for all $i$ and some constant $C^{'}$. We thus obtain
$\sqrt{\sum_{j = 1}^{d} {{U_B}_{ij}}^2} \leq \frac{C}{\sqrt{n}}$,  i.e., $||U_B||_{2 {\to} \infty} \leq \frac{C}{\sqrt{n}}.$
Hence, 
\begin{equation*}
\begin{split}
\| [ U_B U_B^{\top} (\hat{B} - B) U_B W^{*} S_{\hat{B}}^{-1/2}]_{h} \| & \leq \|U_B\|_{2 {\to} \infty} \| U_B^{\top} (\hat{B} - B) U_B \| \times \|S_{\hat{B}}^{-1/2}\| \\
& \leq \frac{C}{\sqrt{n}} \mathcal{O}( \log n) \mathcal{O}(n^{-\frac{1}{2}} ) \leq \frac{C \log n}{n} 
\end{split}
\end{equation*}
which also converges to $0$ as $n \rightarrow \infty$ (note in the last inequality we used \ref{important_bound}).
  
To show Eq.~\eqref{eq:3}, we must bound $\|[ (I -U_B U_B^{\top}) (\hat{B} - B) (\hat{U}_B - U_B W^{*}) S_{\hat{B}}^{-1/2}]_{h}\|$.
Define 
\begin{gather*} G_1 =  (I -U_B U_B^{\top}) (\hat{B} - B) (I -U_B U_B^{\top}) U_{\hat{B}} S_{\hat{B}}^{-1/2}, \\
 G_2 =  (I -U_B U_B^{\top}) (\hat{B} - B) U_B ( U_B^{\top} U_{\hat{B}} - W^{*})S_{\hat{B}}^{-1/2} 
\end{gather*}
Note that $(I -U_B U_B^{\top}) (\hat{B} - B) (\hat{U}_B - U_B W^{*}) {S_{\hat{B}}}^{-1/2} = G_1 + G_2.$ 
We now only need to bound the $h$th row of $G_1$ and $G_2$.
\begin{align*}
    \| G_2 \|_F & \leq \|(I -U_B U_B^{\top}) (\hat{B} - B) U_B \| \times \| U_B^{\top} U_{\hat{B}} - W^{*} \|_F \times \| {S_{\hat{B}}}^{-\frac{1}{2}} \| \\
    & \leq \|(I -U_B U_B^{\top}) \| \times \|\hat{B} - B\| \times \| U_B^{\top} U_{\hat{B}} - W^{*} \|_F \times \| {S_{\hat{B}}}^{-\frac{1}{2}} \| \\
    & = \mathcal{O}(1) \mathcal{O}(\sqrt{n \log n} )  \mathcal{O}(n^{-1} ) \mathcal{O}(n^{-\frac{1}{2}} )  = \mathcal{O}(\frac{\sqrt{\log n}}{n}) 
\end{align*}
Thus $\|\sqrt{n} G_2 \|_F$ converges to $0$ as $n \rightarrow \infty.$
We now consider the rows of $G_1$. Note that $U_{\hat{B}}^{\top} U_{\hat{B}} = I $ and hence
\begin{align*}
    \|(G_1)_h\| & = \| [ (I -U_B U_B^{\top}) (\hat{B} - B) (I -U_B U_B^{\top}) U_{\hat{B}} S_{\hat{B}}^{-1/2} ]_h \| \\
    & = \| [ (I -U_B U_B^{\top}) (\hat{B} - B) (I -U_B U_B^{\top}) U_{\hat{B}} U_{\hat{B}}^{\top} U_{\hat{B}} S_{\hat{B}}^{-1/2} ]_h \| \\
    & = \| U_{\hat{B}} S_{\hat{B}}^{-1/2} \| \times \| [ (I -U_B U_B^{\top}) (\hat{B} - B) (I -U_B U_B^{\top}) U_{\hat{B}} U_{\hat{B}}^{\top} ]_h\| \\
    & \leq \frac{C}{\sqrt{n}} \| [ (I -U_B U_B^{\top}) (\hat{B} - B) (I -U_B U_B^{\top}) U_{\hat{B}} U_{\hat{B}}^{\top} ]_h\|
\end{align*}
Define $$ H_1 = (I -U_B U_B^{\top}) (\hat{B} - B) (I -U_B U_B^{\top}) U_{\hat{B}} U_{\hat{B}}^{\top}.$$ Since the $Z_i$ are i.i.d., the rows of $H_1$ are exchangeable and hence, for any fixed index $h$, $n \mathbb{E} \|(H_1)_h \|^2 = \mathbb{E}[\|H_1\|_F^2]$. Markov's inequality then implies
  \begin{align*}
    \mathbb{P} [ \|\sqrt{n} (H_1)_h \| > t] & \leq \frac{n \mathbb{E} {\| [(I -U_B U_B^{\top}) (\hat{B} - B) (I -U_B U_B^{\top}) U_{\hat{B}} U_{\hat{B}}^\top)_h]\|}^2 }{t^2} \\
    & = \frac{\mathbb{E}\bigl(\| (I -U_B U_B^{\top}) (\hat{B} - B) (I -U_B U_B^{\top}) U_{\hat{B}} U_{\hat{B}}^{\top} \|_F^2\bigr)}{t^2}
  \end{align*}
  Furthermore,
  $$\| (I -U_B U_B^{\top}) (\hat{B} - B) (I -U_B U_B^{\top}) U_{\hat{B}} U_{\hat{B}}^{\top} \|_F \leq \| \hat{B} - B \| \times \|U_{\hat{B}} -U_B U_B^{\top} U_{\hat{B}} \|_F $$
  We now recall the following two observations
  \begin{itemize}
    \item The optimization problem $\min_{T \in \mathbb{R}^{d \times d}} {\| U_{\hat{B}} - U_B T\|_F}^2$ is solved by $T = U_B^{\top} U_{\hat{B}}.$
    \item By theorem 2 of \cite{vD-K}, there exists $W \in \mathbb{R}^{d \times d}$ orthogonal, such that $\| U_{\hat{B}} - U_B W\|_F \leq  
    C \| {U_{\hat{B}}} U_{\hat{B}}^{\top} - U_B U_B^{\top} \|_F.$
  \end{itemize}
  Combining the two facts above, we conclude that
   ${\| U_{\hat{B}} - U_B U_B^{\top} U_{\hat{B}} \|_F}^2  \leq \frac{C}{n}$ with high probability, as in Lemma \ref{appthm9}, hence 
  $$\| (I -U_B U_B^{\top}) (\hat{B} - B) (I -U_B U_B^{\top}) U_{\hat{B}} U_{\hat{B}}^{\top} \|_F \leq \mathcal{O}(\sqrt{n \log n}) \frac{C}{\sqrt{n}} = \mathcal{O}(\sqrt{\log n}),$$
  with high probability. Therefore,
  $$\mathbb{P} (\|\sqrt{n} (H_1)_h \| > t) \leq \frac{\sqrt{\log n}}{t^2}.$$
  picking $t=n^{\frac{1}{4}}$, we get $\lim_{n \rightarrow \infty} Cn^{-1/2} \|\sqrt{n} (H_1)_h\| = 0.$

Finally, Eq.~\eqref{eq:4} and Eq.~\eqref{eq:5} follow from Lemma \ref{appthm8} and Lemma~\ref{appthm9} and the bound $\| U_B\|_{2 \to \infty} \leq C n^{-1/2}$.
\end{proof}

\bibliography{Bib}
\end{document}